\documentclass[12pt, hidelinks]{article}

\usepackage[top=2.8cm, bottom=2.8cm, left=2.8cm, right=2.8cm]{geometry}
\usepackage[english]{babel}

\usepackage{amssymb,amsmath,amsthm,amsfonts,mathtools}
\usepackage{graphicx}
\usepackage{amsfonts}
\usepackage[utf8]{inputenc}
\usepackage{hyperref}
\usepackage{natbib}
\usepackage{datetime}
\usepackage{setspace}
\usepackage{bookmark}
\usepackage{quoting} 
\usepackage{enumerate}
\usepackage{floatrow}
\usepackage{caption}
\theoremstyle{theorem}
\newtheorem{theorem}{Theorem}
\theoremstyle{definition}
\newtheorem{definition}{Definition}
\theoremstyle{remark}

\theoremstyle{corollary}

\theoremstyle{lemma}

\theoremstyle{proposition}

\begin{document}
\baselineskip = 1.4\baselineskip

\title{Social preferences and expected utility}
\author{Mehmet S. Ismail\footnote{Department of Political Economy, King's College London, London, WC2R 2LS, UK.  mehmet.ismail@kcl.ac.uk}
\and Ronald Peeters\footnote{Department of Economics, Otago Business School, University of Otago, Dunedin 9054, New Zealand. ronald.peeters@otago.ac.nz}}

\date{\today}
\maketitle

\begin{abstract}
\baselineskip = 1.4\baselineskip
\noindent
It is well known that ex ante social preferences and expected utility are not always compatible. In this note, we introduce a novel framework that naturally separates social preferences from selfish preferences to answer the following question: What specific forms of social preferences can be accommodated \textit{within} the expected utility paradigm? In a departure from existing frameworks, our framework reveals that ex ante social preferences are not inherently in conflict with expected utility in games, provided a decision-maker's aversion to randomization in selfish utility ``counterbalances'' her social preference for randomization. We also show that when a player's preferences in both the game (against another player) and the associated decision problem (against Nature) conform to expected utility axioms, the permissible range of social preferences becomes notably restricted. Only under this condition do we reaffirm the existing literature's key insight regarding the incompatibility of ex ante inequality aversion with expected utility.


\end{abstract}

\newpage
\section{Introduction}

The impact of social preferences on individual decision-making in economics has been an area of extensive research, initiated by seminal contributions by \citet{fehr1999} and \citet{bolton2000}. These early efforts have primarily focused on deterministic settings, leaving open questions about how social preferences operate under uncertainty---an essential aspect of many economic activities. The literature has subsequently evolved to address social preferences under uncertainty. Notable contributions to this line of research come from both theoretical \citep{diamond1967,machina1989,karni2002,trautmann2009,fudenberg2012,saito2013} and experimental \citep{karni2008,bohnet2008,bolton2010,cettolin2017,miao2018} perspectives. For instance, \citet{saito2013} introduces axioms that give rise to expected inequality averse utility function. However, there remains a gap in the literature regarding how selfish and ex ante social preferences interact \textit{within} the expected utility framework in the sense of \citet{neumann1944}. 

In their seminal paper, \citet{gilboa2003} offer a hint on how to distinguish between social and selfish preferences:
\begin{quote}
``[The dictator] in a dictator game cannot be assumed to treat the outcome ``I get \$90'' as equivalent to ``I chose to take \$90 and to leave \$10 to my opponent.'' Preferences over fairness distinguish the former from the latter ... Such distinctions are precisely about the \textit{difference} between a single-player decision problem and a game.'' \citep[pp.\ 185--186; emphasis added]{gilboa2003}
\end{quote}

\noindent
They go on to axiomatize expected utility in the context of a game, stating:
\begin{quote} 
``We do not claim that the expected utility paradigm is broad enough to encompass all types of psychological or social payoffs.'' (ibid., p.\ 187)
\end{quote}

\noindent
This raises the question: What specific forms of social payoffs can be accommodated within the expected utility paradigm? To address this question and fill the existing gap in the literature, we introduce a novel framework to naturally separate social preferences from selfish preferences, as suggested by the first quotation above. We consider a two-person game and an associated one-person decision problem---i.e., a game against Nature---that is uniquely induced by the game. We design the decision problem in such a way that any variation in the preferences of the player in the two-person game versus the one-person decision problem can be attributed solely to the player's social preferences. This is possible because the two scenarios are identical in all aspects except for the presence of another player in the game. To illustrate, consider the aforementioned example: in one scenario, a decision-maker receives \$90 in a decision problem; in another scenario, she receives \$90 but also allocates \$10 to another player in a game. By holding all other factors constant between the two scenarios, our framework allows for a precise differentiation among various types of functions: selfish utility, social utility, game utility, and material payoff functions.

In a departure from previous frameworks in which ex ante inequality aversion is incompatible with expected utility, our framework reveals a compatibility between any kind of ex ante social preferences and expected utility in games. However, this compatibility comes with specific conditions. For example, we show that consistency between ex ante inequality aversion and expected utility can only emerge if the individual's selfish preferences exhibit a distinct aversion to randomization. Specifically, for a decision-maker with ex ante inequality averse social preferences to be consistent with expected utility, her aversion to randomization in her selfish utility must counterbalance her social preference for randomization.

In our main result, presented in Theorem~\ref{thm:main}, we provide a characterization of the functional form of social utility in relation to selfish utility and game utility. Notably, when a player's preferences in \textit{both} the game and the decision problem satisfy expected utility, significant constraints are imposed on the feasible forms of her social utility. Only under these conditions, do our results corroborate key insights from existing literature \citep{diamond1967,machina1989,fudenberg2012,saito2013}. In particular, ex ante inequality aversion is incompatible with the expected utility function, as shown by \citet{fudenberg2012} in a simple `coin-flip' (two-state) model with equal probabilities, where a weak form of the independence axiom conflicts with ex ante fairness.\footnote{\citet{fudenberg2012} also establish an incompatibility result regarding ex post fairness---our study does not focus on ex post social preferences.} Finally, we observe that (i) the social utility function is unique up to a positive affine transformation, and (ii) a positive affine transformation of both the game utility and selfish utility functions leaves social preferences unchanged.

\section{The setup}

In this section, we formally define the concept of a game and a decision problem that is induced by the game.

Consider two players, Alice (player 1) and Bob (player 2), with their respective sets of pure strategies denoted by $A = \{a_1, a_2, \ldots, a_m\}$ and $B = \{b_1, b_2, \ldots, b_n\}$. The material payoff to player $i\in\{1,2\}$, resulting from a chosen pair of pure strategies, are defined by the function $m_i: A \times B \rightarrow \mathbb{R}$. 

We let $\Delta A$ and $\Delta B$ denote the sets of all probability distributions over $A$ and $B$, respectively. A mixed strategy for Alice and Bob is then represented by $\sigma \in \Delta A$ and $\tau \in \Delta B$, respectively. The material payoffs are extended to the domain of pairs of mixed strategies, with the function $m_i: \Delta A \times \Delta B \rightarrow \mathbb{R}$ denoting player $i$'s \textbf{expected material payoff}. We further define $m$ as the profile $(m_1, m_2)$, such that for any mixed strategy profile $(\sigma,\tau)$, we have $m(\sigma, \tau) = (m_1(\sigma, \tau), m_2(\sigma, \tau))$.

Let $u: (\Delta A \times \Delta B) \times \{g,d\} \rightarrow \mathbb{R}$ denote Alice's utility function, which is \textit{not} necessarily a von Neumann-Morgenstern (vNM) expected utility function.  For each strategy profile $(\sigma, \tau) \in \Delta A \times \Delta B$, $u(\sigma, \tau, g)$ represents Alice's utility in the game with Bob and $u(\sigma, \tau, d)$ represents her utility in the associated decision problem against Nature. We define $u_g(\cdot,\cdot)=u(\cdot,\cdot,g)$ and $u_d(\cdot,\cdot)=u(\cdot,\cdot,d)$ and call them Alice's \textbf{game utility} and \textbf{selfish utility}, respectively. Note that $u(\sigma, \tau, \cdot)$ and $m_1(\sigma, \tau)$ may differ, as material payoffs do not necessarily coincide with utils. Furthermore, let $v: \Delta A \times \Delta B \rightarrow \mathbb{R}$ be Bob's utility function.  

With these definitions, a \textbf{game} in mixed extension is denoted as $G = (\Delta A, \Delta B, u_g,$ $v, m_1, m_2)$, and its \textbf{associated decision problem} is $D = (\Delta A, \Delta B, u_d, m_1)$. There are similarities and differences between the game $G$ and its associated decision problem $D$:
\begin{enumerate}
    \item In both $G$ and $D$, Alice possesses an identical set of mixed strategies.
    \item $D$ is a game against Nature whose set of mixed strategies coincides with mixed strategies of Bob in $G$.
    \item Additionally, Alice receives the same material payoffs, denoted by $m_1$, in both $G$ and $D$. This establishes a key link between $G$ and $D$ and motivates the definition of social utility below.
    \item In contrast, Alice's utilities in $G$ may differ from those in $D$.
\end{enumerate}

\noindent
We proceed to exemplify the above theoretical framework using a game illustrated in Figure~\ref{fig:illustrative_payoffs}. In this game, the material payoffs are measured in monetary units. Note that Alice's material payoffs remain the same whether she is in the game or the associated decision problem. Any divergence in her utility thus arises solely from the nature of her opponent: in the game, she competes against a human player, Bob. This is the only non-material difference between the two scenarios.

\begin{figure}[!htb]
	\[
	\begin{array}{ r|c|c| }
	\multicolumn{1}{r}{}
	&  \multicolumn{1}{c}{\text{Left}}
	& \multicolumn{1}{c}{\text{Right}} \\
	\cline{2-3}
	\text{Left} &  0,20 & 10,0 \\
	\cline{2-3}
	\text{Right}&  30,0 & 0,20   \\
	\cline{2-3}
	\end{array}
	\qquad
	\begin{array}{ r|c|c| }
	\multicolumn{1}{r}{}
	&  \multicolumn{1}{c}{\text{Left}}
	& \multicolumn{1}{c}{\text{Right}} \\
	\cline{2-3}
	\text{Left} & 0 & 10 \\
	\cline{2-3}
	\text{Right}&  30 & 0   \\
	\cline{2-3}
	\end{array}
	\]
	\caption{An illustration of a game (left) and its associated decision problem (right), where numbers represent \textit{monetary payoffs}, not utils.}
	\label{fig:illustrative_payoffs}
\end{figure}

For instance, suppose that Alice has selfish preferences. In such a case, it is reasonable to assume that her selfish utility is equal to her game utility, e.g., $u_g(L,L)=u_d(L,L)$, where the corresponding monetary payoffs are $(0,20)$ and $0$, respectively. However, if Alice is inequality averse, then it would be more appropriate to assume that she prefers the profile $(L,L)$ in the decision problem to the profile $(L,L)$ in the game, given that Bob receives $20$ more monetary units than Alice in the profile $(L,L)$. This indicates that Alice derives greater utility from an outcome where the distribution of material payoffs between her and Bob is more equitable. To capture these types of social preferences under risk, we introduce the following social utility function.

\begin{definition}[Social utility]
    \label{def:social_preferences}
    For each $(\sigma, \tau)\in \Delta A \times \Delta B$, the \textbf{social utility function}, $s: \Delta A \times \Delta B \rightarrow \mathbb{R}$, induced by $u_g$ and $u_d$ is defined as
    \[
    s(\sigma, \tau) = u_g(\sigma, \tau) - u_d(\sigma, \tau).
    \]
\end{definition}

\noindent
The function $s$ is defined as the difference between Alice's game utility and selfish utility. This function measures how her utility changes when engaged in a game with Bob,  thereby capturing her social preferences. Specifically, it quantifies the social utility---that is, any additional utility or disutility---that Alice experiences as a result of the material payoffs $m(\sigma, \tau)$, where Bob's material payoff is $m_2(\sigma, \tau)$ and Alice's is $m_1(\sigma, \tau)$. 

Of note, we define game utility, selfish utility, and social utility as functions of independently chosen mixed strategies (i.e., lotteries). Preferences over these mixed strategy combinations naturally induce preferences over expected monetary outcomes.

\section{Illustrative example}
\label{sec:illustrative}

In this section, we begin by examining Alice's social preferences in the context of lotteries. For illustrative purposes, we do this within the context of the situation depicted in Figure~\ref{fig:illustrative_payoffs}, that presents Alice's material payoffs in the game and the associated decision problem. We define mixed strategies $\widehat{\sigma}=(\frac{1}{2},\frac{1}{2})$ and $\widehat{\tau}=(\frac{1}{2},\frac{1}{2})$ and consider the three mixed strategy profiles $(L,\widehat{\tau})$, $(R,\widehat{\tau})$ and $(\widehat{\sigma},\widehat{\tau})$ illustrated in Figure~\ref{fig:social_preferences}. The expected monetary payoffs at these three profiles are $(5,10)$, $(15,10)$ and $(10,10)$, respectively. 

\begin{figure}[!htb]
	\[
	\begin{array}{ r|c|}
	\multicolumn{1}{r}{}
	& \multicolumn{1}{c}{\widehat{\tau}} \\
	\cline{2-2}
	\text{Left} &  -1  \\
	\cline{2-2}
        \widehat{\sigma} &  0  \\
	\cline{2-2}
	\text{Right}&  -1  \\
	\cline{2-2}
	\end{array}
  \qquad 
  \begin{array}{ r c}
	\multicolumn{1}{r}{}
	& \multicolumn{1}{c}{} \\
	  &    \\
          &  =  \\
	 &    \\
	\end{array}
   \qquad 
	\begin{array}{ r|c|}
	\multicolumn{1}{r}{}
	& \multicolumn{1}{c}{\widehat{\tau}} \\
	\cline{2-2}
	\text{Left} &  4  \\
	\cline{2-2}
        \widehat{\sigma} &  9  \\
	\cline{2-2}
	\text{Right}&  14  \\
	\cline{2-2}
	\end{array}
  \qquad 
    \begin{array}{ r c}
	\multicolumn{1}{r}{}
	& \multicolumn{1}{c}{} \\
	  &    \\
          &  -  \\
	 &    \\
	\end{array}
   \qquad 
	\begin{array}{ r|c|}
	\multicolumn{1}{r}{}
	& \multicolumn{1}{c}{\widehat{\tau}} \\
	\cline{2-2}
	\text{Left} &  5  \\
	\cline{2-2}
        \widehat{\sigma} &  x  \\
	\cline{2-2}
	\text{Right}&  15  \\
	\cline{2-2}
	\end{array}
 \]
	\caption{Values from Alice's social utility (left), game utility (middle), and selfish utility (right). These values represent \textit{utils}. In line with Definition~\ref{def:social_preferences}, equality and subtraction operations are performed component-wise.}
	\label{fig:social_preferences}
\end{figure}

Let Alice be ex ante inequality averse in that she incurs a social disutility of $1$ only if expected monetary payoffs are unequal. Thus, she experiences a disutility of $1$ from the profiles $(L,\widehat{\tau})$ and $(R,\widehat{\tau})$, but does not incur any social disutility from the profile $(\widehat{\sigma},\widehat{\tau})$. Her social utility at these three profiles is then given by 
\[
s(L,\widehat{\tau}) = -1,\quad s(R,\widehat{\tau}) = -1 \quad\textrm{and}\quad s(\widehat{\sigma},\widehat{\tau}) = 0.
\]

\noindent
In the decision problem, we assume Alice's selfish utility values for the two profiles where she responds to $\widehat{\tau}$ with a pure strategy are
\[
u_d(L,\widehat{\tau}) = 5 \quad\textrm{and}\quad u_d(R,\widehat{\tau}) = 15.
\]

\noindent
By the assumptions we made so-far, and by the definition of $s$, this determines the game utilities at the latter two profiles as 
\[
u_g(L,\widehat{\tau}) = s(L,\widehat{\tau}) + u_d(L,\widehat{\tau}) = -1 + 5 = 4
\]
and
\[
u_g(R,\widehat{\tau}) =  s(R,\widehat{\tau})+ u_d(R,\widehat{\tau}) = -1 + 15 = 14.
\]

\noindent
Now, assuming Alice's game utility is vNM expected utility implies that
\[
\textstyle u(\widehat{\sigma},\widehat{\tau}) = \frac{1}{2} u_g(L,\widehat{\tau}) + \frac{1}{2} u_g(R,\widehat{\tau}) = 9.
\]

\noindent
Consequently, from the definition of $s$, we get
\[
\textstyle s(\widehat{\sigma},\widehat{\tau}) = u_g(\widehat{\sigma},\widehat{\tau}) - u_d(\widehat{\sigma},\widehat{\tau}) = 9 - x.
\]
Since, $s(\widehat{\sigma},\widehat{\tau})=0$, this implies that $x=9$, and hence that 
\[
u_d(\widehat{\sigma},\widehat{\tau})=9.
\]

\noindent
In summary, despite having ex ante inequality averse social preferences, Alice's game utility function can be a vNM expected utility function, provided her selfish utility from the profile $(\widehat{\sigma},\widehat{\tau})$ is $9$. In other words, Alice's social preference for randomization in the game must be counterbalanced by her selfish dislike of randomization in the decision problem.

However, if Alice's selfish utility were also expected utility, then 
\[
\textstyle u_d(\widehat{\sigma},\widehat{\tau}) = \frac{1}{2} u_d(L,\widehat{\tau}) + \frac{1}{2} u_d(L,\widehat{\tau}) = 10.
\]
This would imply that her game utility for the profile $(\widehat{\sigma},\widehat{\tau})$ would have to be $10$, contradicting to the assumption that her game utility $u_g$ is a vNM expected utility function.

\section{Result and discussion}
\label{sec:results}

As our main theorem establishes, the restriction expected utility functions imposes on social preferences is not merely an idiosyncratic feature of the particular illustrative example. Rather, it points to a more general phenomenon: the expected utility functions in the game and in the decision problem---when considered together---impose significant constraints on the form that social utility can take.

\begin{theorem}
\label{thm:main}
Let $u_g: \Delta A \times \Delta B \rightarrow \mathbb{R}$ be a vNM expected utility function, $u_d: \Delta A \times \Delta B \rightarrow \mathbb{R}$ a selfish utility function, and $s$ a social utility function induced by $u_g$ and $u_d$. Then, $u_d$ is a vNM expected utility function if and only if $s$ is bilinear.
\end{theorem}

\begin{proof}
\textbf{($\Rightarrow$)} 
By the definition of $s$, for each pair $(\sigma, \tau)$ in $\Delta A \times \Delta B$, we have
\begin{align*}
s(\sigma, \tau) &= u_g(\sigma, \tau) - u_d(\sigma, \tau) \\
& \stackrel{(1)}{=} \sum_{i=1}^{m} \sum_{j=1}^{n} \sigma(a_i) \tau(b_j) u_g(a_i, b_j) - \sum_{i=1}^{m} \sum_{j=1}^{n} \sigma(a_i)\tau(b_j) u_d(a_i, b_j) \\
& \stackrel{(2)}{=} \sum_{i=1}^{m} \sum_{j=1}^{n} \sigma(a_i) \tau(b_j) u_g(a_i, b_j) - \sum_{i=1}^{m} \sum_{j=1}^{n} \sigma(a_i)\tau(b_j) (u_g(a_i, b_j)-s(a_{i}, b_{j})) \\
& \stackrel{(3)}{=} \sum_{i=1}^{m} \sum_{j=1}^{n} \sigma(a_i) \tau(b_j) s(a_{i}, b_{j}).
\end{align*}
Here, Equality~(1) uses the assumption that $u_g$ and $u_d$ are vNM expected utility functions. Next, Equality~(2) follows from the definition of $s$. Finally, Equality~(3) follows from the distributive property of multiplication and cancelling out equal terms.

\textbf{($\Leftarrow$)} 
Definition of $s$ implies that 
\begin{align*}
u_d(\sigma, \tau) &=  u_g(\sigma, \tau) - s(\sigma, \tau) \\
& \stackrel{(1)}{=} \sum_{i=1}^{m} \sum_{j=1}^{n} \sigma(a_i)\tau(b_j) u_g(a_i, b_j) - \sum_{i=1}^{m} \sum_{j=1}^{n} \sigma(a_i) \tau(b_j) s(a_i, b_j) \\
& \stackrel{(2)}{=} \sum_{i=1}^{m} \sum_{j=1}^{n} \sigma(a_i)\tau(b_j) u_g(a_i, b_j) - \sum_{i=1}^{m} \sum_{j=1}^{n} \sigma(a_i) \tau(b_j) (u_g(a_i, b_j)-u_d(a_i, b_j)) \\
& \stackrel{(3)}{=} \sum_{i=1}^{m} \sum_{j=1}^{n}  \sigma(a_i) \tau(b_j) u_d(a_i, b_j).
\end{align*}
Here, Equality~(1) is obtained by our assumptions that $u_g$ is a vNM utility function and that $s$ is bilinear. Analogous to the above case, Equality~(2) follows from the definition of $s$, and Equality~(3) follows from the distributive property of multiplication and the cancellation of equal terms.
\end{proof}

\noindent
From a mathematical perspective, the theorem might seem straightforward or even trivial. However, in an economic context, the properties of function $s$ are instrumental as $s$ quantifies Alice's social preferences. Our result suggests that ex ante social preferences are not inherently in conflict with expected utility in a game, provided that selfish preferences counterbalance social preferences. When both game utility and selfish utility are vNM expected utility functions, the permissible range of social preferences becomes notably restricted. Specifically, any preferences towards more equitable expected payoffs are precluded. Under these particular conditions, our result echoes key insights from existing literature. For example, it indicates that ex ante inequality aversion would not be compatible with expected utility.

We also would like to discuss whether scaling functions $s$ or $u$ changes the underlying social preferences. Suppose that $\widehat{u}$ and $\widehat{s}$ are positive affine transformations of $u$ and $s$, respectively. Then, it is straightforward to show that for every $ (\sigma, \tau)$ and $ (\sigma', \tau')$, (i) $\widehat{s}(\sigma, \tau)\geq \widehat{s}(\sigma', \tau')$ implies that $s(\sigma, \tau)\geq s(\sigma', \tau')$ and (ii) $\widehat{u}_g(\sigma, \tau) - \widehat{u}_d(\sigma, \tau)\geq \widehat{u}_g(\sigma', \tau') - \widehat{u}_d(\sigma', \tau')$ implies that $s(\sigma, \tau)\geq s(\sigma', \tau')$. In summary, this observation states that the social utility function is unique up to a positive affine transformation and that a positive affine transformation of $u$ leaves social preferences unchanged. 

An important point to note here is that both $u_g$ and $u_d$ are scaled using identical parameters through a scaling of $u$. If $u_g$ and $u_d$ were to be scaled by different parameters, part (ii) of the observation would no longer hold. Therefore, it is essential to assume that both $u_g$ and $u_d$ are derived from a common utility function $u$ to maintain consistency between social and selfish preferences when the functions representing them are scaled. In contrast, the conclusion of Theorem~\ref{thm:main} remains valid even if $u_g$ and $u_d$ are unrelated expected utility functions, permitting them to be scaled separately.

\newpage
\bibliographystyle{chicago}
\bibliography{references}

\end{document}